\documentclass{article}

\usepackage{times}
\usepackage{graphicx} 
\usepackage{subfigure} 

\usepackage{natbib}

\usepackage{amsthm,amsfonts}

\usepackage{comment}

\usepackage{bbm,amsmath,amssymb,natbib,url}

\usepackage{algorithm}
\usepackage{algorithmic}

\usepackage{hyperref}

\usepackage[accepted]{icml2016} 

\newtheorem{Thm}{Theorem}[section]
\newtheorem{Lem}[Thm]{Lemma}
\newtheorem{Cor}[Thm]{Corollary}

\newtheorem{Def}[Thm]{Definition}

\icmltitlerunning{Approximate MAP Estimation for Pairwise Potentials via Baker's Technique}

\begin{document} 

\twocolumn[
\icmltitle{Approximate MAP Estimation for Pairwise Potentials via Baker's Technique}

\icmlauthor{Yi-Kai Wang}{wangyikai1029@gmail.com}

\icmlkeywords{Graphical Model, MAP estimation}

\vskip 0.3in
]

\begin{abstract} 
Graphical models with pairwise correlations between variables are widely used to model optimization problems in machine learning and other fields. The structures of these optimization problems can be encoded as potential functions attached on the vertices of the input graph. Then maximum a posterior (MAP) estimation is equivalent to maximizing or minimizing the energy function, i.e. sum of the potential functions. 
We show that if the potentials are nonnegative, then maximizing the energy admits efficient polynomial-time approximation schemes (EPTAS) on planar graphs, bounded-local-treewidth graphs, $H$-minor-free graphs and bounded-crossing-number graphs.
Our EPTAS can be applied to various significant optimization problems in machine learning, data mining, computer vision, combinatorial optimization and statistical physics.
We also prove that approximation algorithm does not exist for minimization even if the potentials are nonnegative and the input graph is planar.
Our method is a simple extension of \textit{Baker's Technique} and consequently it also generalizes a series of related works proposed over the last three decades.
\end{abstract}

\section{Introduction}
\label{section:introduction}

\subsection{The Model}
\label{subsection:model}

Graphical models with pairwise potentials are widely used for research of machine learning. Its maximum a posterior (MAP) estimation plays a key role for many learning-as-optimization tasks.
The structures of several significant optimization problems in machine learning and some other fields can be encoded as potential functions attached on the vertices and edges of a graph whose vertices represent variables and edges represent correlations between the variables. 
Formally, the input graph is denoted by $G(V,E)$ where $V$ denotes the set of vertices and $E$ denotes the set of edges. 
Each vertex $i \in V$ represents a variable $\sigma_i \in [q]$ that can take $q$ different values and is attached
by a \textit{vertex potential function} $\psi_i: [q]\to\mathbb{R}$. 
Each edge $(i,j) \in E$ is attached by an \textit{edge potential function} $\psi_{ij}: [q]\times[q]\to\mathbb{R}$ which takes the values of $\sigma_i$ and $\sigma_j$ as inputs. A joint value assignment $\sigma \in [q]^V$ is called a \textit{configuration}. The \textit{energy function} $E(\sigma)$ is defined as
$E(\sigma) = \sum_{i \in V}\psi_i(\sigma_i)+\sum_{(i,j)\in E}\psi_{ij}(\sigma_i,\sigma_j)$.
The following well-known optimization problems can be reduced to finding a configuration $\sigma$ maximizing or minimizing $E(\sigma)$. 

\noindent\textbf{MAP Estimation}: 
Maximum a posterior (MAP) estimation on graphical models is a fundamental problem in machine learning. Given a pairwise Markov random field, the Gibbs measure of configuration $\sigma \in [q]^V$ is usually defined as $P(\sigma) = \frac{1}{Z}\exp(-E(\sigma))$ where $Z = \sum_{\sigma \in [q]^V}\exp(-E(\sigma))$ is the \textit{partition function}. The goal of MAP estimation is finding the configuration $\sigma$ that maximizes $P(\sigma)$. This is equivalent to minimizing the energy function $E(\sigma)$. The definition of Gibbs measure follows the convention of statistical physics. In many cases, we also let $P(\sigma) = \frac{1}{Z}\exp(E(\sigma))$ where $Z = \sum_{\sigma \in [q]^V}\exp(E(\sigma))$, then MAP estimation is equivalent to maximizing the energy function.
This problem is exactly solvable if the input graph is an acyclic graph.
In general, it is NP-hard and few provable approximation bounds have been achieved.

\noindent\textbf{Correlation Clustering}:
Correlation clustering, motivated by document clustering and agnostic learning, provides a method for partitioning data points into clusters based on their similarities. 
It has been commonly used in machine learning and data mining.
The model originally proposed in~\cite{bansal2004correlation} is a complete graph $K_n$ whose vertices represent data points. Each edge has weight either $+1$ (similar) or $-1$ (different) to measure the similarity  of two vertices.
The solutions consists of two scenarios: maximizing agreements (maximizes sum of positive weights in clusters plus sum of absolute values of negative weights between clusters) or minimizing disagreements (minimizes sum of absolute values of negative weights in clusters plus sum of positive weights between clusters).
Both of them are NP-complete.
In~\cite{bansal2004correlation}, a PTAS is given for maximizing agreements and a constant factor approximation algorithm is given.
For general graphs $G(V,E)$ with real-valued weights $w: E \to \mathbb{R}$, an $O(\log n)$-approximation is given for minimizing disagreements and a 0.7664-approximation is given for maximizing agreements in~\cite{charikar2003clustering}. It is also proved that maximizing agreements is APX-hard and minimizing disagreements is APX-hard on complete graphs. 
Later in~\cite{swamy2004correlation}, a 0.766-approximation algorithm is given for maximizing agreements via semidefinite programming. The approximation ratio 0.766 also holds for $q$-clustering variant where the number of clusters is at most $q$.
Let $\psi_i = 0$.  
If $w_{ij} > 0$, we set $\psi_{ij} = w_{ij}$ if $\sigma_i = \sigma_j$ and $\psi_{ij} = 0$ otherwise. 
If $w_{ij} < 0$, we set $\psi_{ij} = -w_{ij}$ if $\sigma_i \ne \sigma_j$ and $\psi_{ij} = 0$ otherwise. 
Then maximizing agreements for $q$-clustering is equivalent to maximizing the energy function.
Minimizing disagreements for $q$-clustering can be reduced to minimizing the energy function via a similar reduction.

\noindent\textbf{MAX Graph-cuts}: 
Given an undirected graph $G(V,E,w)$ where $w:E \to \mathbb{R}^+$ assign a nonnegative weight $w_{ij}$ to each edge $(i,j) \in E$. The goal of MAX-CUT problem is dividing the vertices of $G$ into two sets $S$ and $\bar{S}$ such that the value $\sum_{(i,j) \in C(S,\bar{S})} w_{ij}$ is maximum where $C(S,\bar{S})$ is the set of cut edges between $S$ and $\bar{S}$.
Its directed-graph version is MAX-DICUT problem whose input is a directed graph $G$ and whose goal is dividing the vertices of $G$ into two sets $S$ and $\bar{S}$ such that the total weight of the directed cut $\sum_{i \in S, j \in \bar{S}, (i,j) \in E} w_{ij}$ is maximum. 
Let $\psi_i = 0$.
For MAX-CUT, we let $\psi_{ij}(\sigma_i,\sigma_j) = w_{ij}$ if $\sigma_i \ne \sigma_j$ and $\psi_{ij}(\sigma_i,\sigma_j) = 0$ otherwise. For MAX-DICUT, we let $\psi_{ij}(\sigma_i,\sigma_j) = w_{ij}$ if $\sigma_i = 1 \wedge \sigma_j = 0$ and $\psi_{ij}(\sigma_i,\sigma_j) = 0$ otherwise. Then computing the maximum cuts is equivalent to maximizing the corresponding energy function.
The best known approximation ratio for MAX-CUT is $\alpha = \frac{2}{\pi}\min_{0 \le \theta \le \pi}\frac{\theta}{1-\cos \theta} \approx 0.878$ discovered in~\cite{goemans1995improved} using semidefinite programming and randomized rounding.
In~\cite{khot2007optimal}, it is shown that this is the best possible approximation ratio for MAX-CUT if the unique game conjecture~\cite{khot2002power,khot2005unique} is true. 
In~\cite{barahona1988application}, minimizing the number of vias (holes on a printed circuit board) for very-large-scale-intergrated (VLSI) circuit design is reduced to computing MAX-CUT.

\noindent\textbf{Statistical Physics}: 
Spin system is a theoretical model for studying properties like ferromagnetism and phase transition.
Edwards-Anderson model is a widely accepted description of the spin systems. 
The input graph is usually a $d$-dimensional lattice graph $\mathbb{L}^d$. Each vertex $i$ in the lattice is a Ising spin $\sigma_i \in \{-1, +1\}$.
The energy functions is $E(\sigma) = \sum_{(i,j) \in E}J_{ij}\sigma_i\sigma_j+B\sum_{i \in V}\sigma_i$ where $J_{ij}$ are exchange couplings and the second part of the sum represents the external magnetic field.
The interactions between spins $\sigma_i$ and $\sigma_j$ is ferromagnetic if $J_{ij} > 0$ and antiferromagnetic if $J_{ij} < 0$.
When $J_{ij} > 0$ and $B = 0$, the energy function becomes $E(\sigma) = \sum_{(i,j) \in E}J_{ij}\sigma_i\sigma_j = C - 2\sum_{(i,j) \in E,\sigma_i \ne \sigma_j}J_{ij}$.
Then computing the ground state is equivalent to computing the maximum weighted cut of the input graph.
This problem is NP-hard even on 3-dimensional lattice graphs~\cite{mezard2009information}.
Techniques such as branch and bound methods, belief propagation have been applied but few provable bounds have been achieved.

\noindent\textbf{Computer Vision}: 
The theoretical models of statistical physics are widely used in computer vision. The pixel values are analogous to states of atoms in a lattice-like spin system.
The $\psi_i(\cdot)$ part of the energy function measures the disagreement between the label and observed value at pixel $i$.
The $\psi_{ij}(\cdot,\cdot)$ part measures the pairwise smoothness between pixel $i$ and pixel $j$.
The likelihood of a label configuration $\sigma$ is measured by the probability $P(\sigma) = \frac{1}{Z(\beta)}\exp(-\beta E(\sigma))$ where $\beta$ is a free parameter called \textit{inverse temperature} and $Z(\beta) = \sum_{\sigma \in [q]^V}\exp(-\beta E(\sigma))$ is the normalizing factor.
The typical $\psi_i(\cdot)$ usually takes the form $\psi_i(\sigma_i) = (\sigma_i - p_i)^2$~\cite{boykov2001fast} where $p_i$ is the observed value of pixel $i$. The typical forms of $\psi_{ij}(\sigma_i,\sigma_j)$ are~\cite{boykov1998markov} generalized Potts model that $\psi_{ij}(\sigma_i,\sigma_j) = w_{ij}\cdot(1-\delta(\sigma_i, \sigma_j))$ where $w_{ij} \ge 0$ is a weight coefficient and $\delta(x) = 1$ when $x = 0$, and $\delta(x) = 1$ otherwise. 
Then computing the configuration with maximum likelihood is equivalent to minimizing the energy function.
This method has been widely used in various applications such as image restoration, image segmentation, texture synthesis and stereo vision.

\noindent\textbf{MAX 2-CSP}: 
Constraint satisfaction problems (CSP) are well investigated in artificial intelligence.
A CSP is defined as a tuple $\langle X, Q, C \rangle$ where $X = \{X_1,\ldots,X_n\}$ is a set of variables, $Q = \{Q_1,\ldots,Q_n\}$ is the set of respective value domains of $X$, $C = \{C_1,\ldots,C_m\}$ is a set of constraints. The variable $X_i$ can take the values in domain $Q_i$. Each constraint $C_\ell$ is a function which takes a subset $S_\ell \subseteq X$ of variables as inputs and returns a number $w_\ell \ge 0$ representing the weight of $C_\ell$ if it is satisfied or 0 otherwise. The most extensively studied CSP is its boolean version SAT.
A MAX-CSP asks for a configuration of variables such that the number of satisfied constraints is maximum. If each constraint takes exactly two variables as inputs, this problem is called MAX 2-CSP.
Although 2-SAT can be decided in polynomial time, MAX-2SAT is APX-hard~\cite{ausiello1999complexity}.
Given a MAX 2-CSP, we construct a graph $G(V,E)$ where each $X_i$ corresponds to a vertex in $V$ and $(i,j) \in E$ if $\exists \ell$ s.t. $S_\ell = \{X_i, X_j\}$.
We set $\psi_i = 0$ and $\psi_{ij}$ as
$\psi_{ij}(\sigma_i,\sigma_j) = \sum_{\ell:S_\ell = \{X_i, X_j\}} C_\ell$,
then solving this MAX 2-CSP is equivalent to maximizing the energy function of $G$.


For each function $\psi_{ij}$ attached on edge $(i,j) \in E$, let $\alpha_{ij}$ and $\alpha_{ji}$ be two constants satisfying 
$\alpha_{ij} + \alpha_{ji} = 1$.
We set $f_i = \psi_i + \sum_{j \in N(i)} \alpha_{ij} \cdot \psi_{ij}$, then computing the energy function $E(\sigma)$ is equivalent to computing $\sum_{i \in V} f_i$.
We do this because our approximation algorithm holds for $f_i \ge 0$, which is necessary but not sufficient for $\psi_i \ge 0$ and $\psi_{ij} \ge 0$.
By this transformation, our approximation algorithm can be applied to a much larger domain of inputs that allows some $\psi_i$ and $\psi_{ij}$ to be negative.
In this paper, we assume that all $f_i$ are $O(1)$-time computable.

\subsection{Main Results}

A polynomial-time approximation scheme (PTAS) is an algorithm $\mathcal{A}(I,\epsilon)$ which takes an instance $I$ of an optimization problem and a parameter $\epsilon > 0$ and runs in time $n^{O(f(1/\epsilon))}$ which produces a solution that is at least $(1-\epsilon)$ optimal for maximization and at most $(1+\epsilon)$ optimal for minimization. A PTAS with running time $f(1/\epsilon) \cdot n^{O(1)}$ is called an efficient polynomial time approximation scheme (EPTAS). An EPTAS where $f(1/\epsilon)$ is polynomial in $1/\epsilon$ is a fully polynomial time approximation scheme (FPTAS).

Our contributions are two-fold. One positive result of efficient approximation algorithm for maximizing $\sum_{i \in V} f_i$ and one negative result of inapproximability property for minimizing $\sum_{i \in V} f_i$. The positive one is given as follows.

\begin{Thm}\label{thm:main-energy}
If the functions $f_i$ derived from $\psi_i$ and $\psi_{ij}$ satisfies $f_i \ge 0$ for all $i \in V$, then computing $\max_{\sigma \in [q]^V}\sum_{i \in V} f_i$ admits EPTASs on planar graphs, bounded-local-treewidth graphs, $H$-minor-free graphs and bounded-crossing-number graphs.
\end{Thm}

The time complexity of our EPTAS is $O\left(q^{O\left(\frac{1}{\epsilon}\right)}\cdot \frac{n}{\epsilon}\right)$. Then we have the following corollary.

\begin{Cor}
Given a fixed error $0<\epsilon<1$, if the functions $f_i$ derived from $\psi_i$ and $\psi_{ij}$ satisfies $f_i \ge 1$ for all $i \in V$, then computing the max-product can be approximated to at least 
$\left(\max_{\sigma \in [q]^V}\prod_{i \in V}\psi_i\prod_{(i,j) \in E}\psi_{ij}\right)^{1-\epsilon}$
with time complexity $O\left(q^{O\left(\frac{1}{\epsilon}\right)}\cdot \frac{n}{\epsilon}\right)$
on planar graphs, bounded-local-treewidth graphs, $H$-minor-free graphs and bounded-crossing-number graphs.
\end{Cor}
To our knowledge, such provable bound for computing the max-product have not been known for other methods such as belief propagation and its variant versions.

By the reductions listed in Section~\ref{subsection:model}, we also have the following corollary.
\begin{Cor}\label{cor:main-csp}
Computing MAX 2-CSP, MAX-CUT, MAX-DICUT, MAX $k$-CUT, maximizing agreements for $q$-clustering and computing the ground state of ferromagnetic Edwards-Anderson model without external magnetic field admits EPTASs on planar graphs, graphs with bounded local treewidth, $H$-minor-free graphs and bounded-crossing-number graphs.
\end{Cor}

On planar graphs, MAX-CUT is polynomial-time solvable~\cite{hadlock1975finding}.
The PTAS for MAX-CUT on $H$-minor-free graphs is given in~\cite{demaine2005algorithmic}. 
MAX $k$-CUT is a natural generalization of MAX-CUT. 
Given a connected undirected graph $G(V,E,w)$ where $w:E \to \mathbb{R}^+$ assigns a nonnegative weight $w_{ij}$ to each edge $(i,j) \in E$, a $k$-cut is a set of edges $E' \subseteq E$ whose removal decomposes the input graph into $k$ disjoint subgraphs.
The goal of MAX $k$-CUT problem is computing such a set of edges $E'$ that $\sum_{e \in E'} w_e$ is maximum. 
When $k = 2$, MAX $k$-CUT problem is MAX-CUT. If $k > \Delta$ where $\Delta$ is the maximum degree of $G$, the optimal solution is precisely the sum of all the weights of the edges. Thus this problem is only interesting when $k \le \Delta$.
Choosing a set of terminals $S = \{s_1, \ldots, s_k\} \subseteq V$, the configuration of $s_t$ ($1 \le t \le k$) is fixed to $t \in [q]$.
The vertices in graph $G' = G \backslash S$ are free variables. The functions attached on $(i,j) \in E$ are the same as those of MAX-CUT.
Then computing the MAX $k$-CUT for fixed $S$ is equivalent to computing $\max_{\sigma \in [k]^{V(G')}}\sum_{(i,j) \in E} \psi_{ij}$. By Theorem~\ref{thm:main-energy}, we have a EPTAS to compute MAX $k$-CUT for fixed $S$. There are at most $P(n,k) = \frac{n!}{(n-k)!} \le n^k$ possibilities for $S$, which is polynomial if $k$ is fixed. By enumerating all these cases, we have a EPTAS for MAX $k$-CUT.

Unfortunately, we have the following inapproximability result for minimization.
\begin{Thm}
Even if $f_i \ge 0$ for all $i \in V$, there does not always exist PTAS for minimizing the energy function even on planar graph unless P = NP.
\end{Thm}
We construct a reduction from computing the chromatic number to energy minimization. We let
\begin{equation*}
f_i =
\begin{cases}
w_x & \text{ if } \sigma_i = x \in [q]
\text{ and } \sigma_i \ne \sigma_j \text{ for } \forall (i,j) \in E \\
+\infty & \text{ if } \sigma_i = \sigma_j \text{ for any } (i,j) \in E
\end{cases}
\end{equation*}
where $w_x$ satisfies $w_x > n \cdot w_{x-1}$.
This is allowed because both $G$ and the set of functions $f_i$ serve as inputs. 
If the graph $G$ is $x$-colorable, then the value of the min-sum $S_{\min} = \min_{\sigma}\sum_{i \in V}f_i$ must fall into the interval $[w_x, w_x \cdot n]$. For any $x,y \in [q]$, $[w_x, w_x \cdot n]$ and $[w_y, w_y \cdot n]$ are pairwise disjoint.
By four color theorem, any planar graph is 4-colorable. Furthermore, it is known that 3-coloring problem remains NP-complete even on planar graph of degree 4~\cite{dailey1980uniqueness}. It implies that it is NP-hard to approximate the chromatic number within approximation ratio $4/3$ even on planar graphs.
Therefore if we have a PTAS for computing the min-sum of $f_i$ on planar graphs, then we have a polynomial time algorithm for computing the chromatic number of planar graphs, which leads to a contradiction.
It implies that there also does not exist PTAS for computing the min-sum on other classes of graphs we will discuss.
Actually, this proof is not only about PTAS. If we set $w_x > g(n) \cdot n \cdot w_{x-1}$, then there does not exist $g(n)$-approximation algorithm for minimizing the energy function.
By enlarging the gaps between the disjoint intervals, we can achieve stronger inapproximability results.

\subsection{Content Organization}

This paper is organized as follows. In Section~\ref{section:preliminary}, we introduce some basic concepts that will be used throughout this paper. In Section~\ref{section:baker}, we gives a concise description of Baker's technique and an overview of the proof of Theorem~\ref{thm:main-energy}.
In Section~\ref{section:decomposition}, we introduce graph decomposition techniques for graph classes mentioned in Theorem~\ref{thm:main-energy}.
Combining the content of Section~\ref{section:baker} and Section~\ref{section:decomposition} will give a complete proof of Theorem~\ref{thm:main-energy}.

\section{Preliminaries}
\label{section:preliminary}


\subsection{Planar graph}

A graph $G$ is a planar graph if it can be embedded into the two-dimensional plane such that no pair of edges will cross with each other.
Given a planar graph, its planar embedding can be generated in linear time. 
A planar graph is an outerplanar graph if it has a planar embedding where all the nodes are on the exterior face. 
Given a planar embedding of a planar graph, a node is at level 1 if it is on the exterior face. 
When all the level-1 nodes are deleted from the planar embedding, the nodes on the exterior face are called level-2 nodes. 
By this induction the level-$k$ nodes can be defined. 
A planar graph is a $k$-outerplanar graph if it has a planar embedding with no nodes of level more than $k$.

\subsection{Tree decomposition and treewidth}

The concept of tree decomposition is defined to measure the similarity between a graph and a tree.
\begin{Def}\label{def:tree-decom}
A tree decomposition of an undirected graph $G(V,E)$ is a tuple $(\{X_i|i\in I\},T=(I,F))$ where $\{X_i|i\in I\}$ is a family of subsets of $V$ that each one corresponds to a node of $T$. $T$ is a tree such that
(1) $\bigcup_{i \in I} X_i = V$,
(2) for all edges $\{v,w\} \in E$, there exists an $i \in I$ with $v \in X_i$ and $w \in X_i$,
(3) for all $i,j,k \in I$: if $j$ is on the path from $i$ to $k$ in $T$, then $X_i \cap X_k \subseteq X_j$.
\end{Def}
Each node of the tree decomposition $T$ is called a \textit{bag}.
The third property of tree decomposition guarantees that for every $v \in V$, $\left\{ X_i: v \in X_i, i \in I \right\}$ induces a connected subtree of $T$.

\begin{Def}
The treewidth of a tree decomposition $(\{X_i|i \in T\},T=(I,F))$ is $\max_{i \in I}|X_i|-1$. The treewidth of a graph $G$, denoted by $tw(G)$, is the minimum treewidth over all tree decompositions of $G$.
\end{Def}
A tree decomposition of width equal to the treewidth is called an optimal tree decomposition. Computing the treewidth for graph $G$ is NP-complete. But given a graph $G$, deciding whether the treewidth of $G$ is at most a fixed constant $k$ can be decided in linear time by Bodlaender's algorithm~\cite{bodlaender1993linear}. If the answer is yes, then an optimal tree decomposition of $G$ can be constructed in linear time (but exponential in $k$).
The following lemma includes some well-known facts about treewidth.
\begin{Lem}\label{lem:treewidth}
Let $(\{X_i|i\in I\},T=(I,F))$ be a tree decomposition of graph $G$. Then
(1) If $X \subseteq V(G)$ is a clique, then there is an $i \in I$ that $X \subseteq X_i$.
(2) Let $G,H$ be graphs such that $V(G) \cap V(H)$ is a clique in both $G$ and $H$. Then it holds that $tw(G \cup H) = \max\{tw(G),tw(H)\}$.
(3) For any $X \subseteq V(G)$. Then $tw(G) \le tw(G \backslash X) + |X|$.
(4) Let $G,H$ be graphs such that $H \le_m G$. Then $tw(H) \le tw(G)$.
\end{Lem}

\subsection{Local treewidth}

The concept of local treewidth is first introduced by~\cite{eppstein2000diameter} as a generalization of treewidth. The local treewidth of graph $G$ is a function that maps an integer $r \in \mathbb{N}$ to the maximum treewidth of the subgraph of $G$ induced by the $r$-neighborhood $N_r(i)$ of any vertex $i \in V$, formally defined as follows.

\begin{Def}
The local treewidth of graph $G(V,E)$ is a function defined as $ltw^G(r) = \max\{tw(G[N_r(i)]): i \in V\}$
where $G[N_r(i)]$ is the subgraph of $G$ induced by $N_r(i)$.
\end{Def}

\begin{Def}
A class $\mathcal{C}$ of graphs has bounded local treewidth if there is a function $f: \mathbb{N} \to \mathbb{N}$ such that $ltw^G(r) \le f(r)$ for all $G \in \mathcal{C}$, $r \in \mathbb{N}$. $\mathcal{C}$ has linear local treewidth if there is a $\lambda \in \mathbb{R}$ such that $ltw^G(r) \le \lambda \cdot r$ for all $G \in \mathcal{C}$, $r \in \mathbb{N}$.
\end{Def}

\subsection{Graph minor}

Given a graph $G$, if graph $H$ can be reduced from a subgraph of $G$ by a sequence of edge contractions, then $H$ is a \textit{minor} of $G$, denoted by $H \le_m G$. We can see that $H \le_m G$ if and only if there is a mapping $h: V_H \to 2^{V_G}$ such that for all $x \in V_H$ the subgraph $G[h(x)]$ of $G$ induced by $h(x)$ is a connected,
$h(x) \cap h(y) = \emptyset$ for all $x \ne y \in V_H$ and, for every $(x,y) \in E_H$ there exists an edge $(u,v) \in E_G$ such that $u \in h(x)$ and $v \in h(y)$.

A class $\mathcal{C}$ of graphs is minor-closed if and only if for all $G \in \mathcal{C}$ and $H \le_m G$ we have $H \in \mathcal{C}$. We say $\mathcal{C}$ is nontrivial if $\mathcal{C}$ does not contain all the graphs. A class $\mathcal{C}$ of graphs is $H$-minor-free if $H \nleq_m G$ for all $G \in \mathcal{C}$. Then we call $H$ an excluded minor of $\mathcal{C}$. Robertson and Seymour's Graph Minor Theorem~\cite{robertson2004graph}, which solves Wagner's conjecture, demonstrates that the undirected graphs partially ordered by the graph minor relationship form a well-quasi-ordering. This implies that every minor-closed class of graphs can be characterized by a finite set of forbidden minors.

It is well known that the treewidth of a $k \times k$ grid is $k$, so planar graphs do not have bounded treewidth. 
But planar graphs have bounded local treewidth. A $k$-outerplanar graph has treewidth at most $3k-1$~\cite{bodlaender1986classes}.
It is also well known that graphs embeddable on bounded-genus surface have bounded local treewidth~\cite{eppstein2000diameter}.
The following theorem gives a precise characterization of graphs with bounded local treewidth.
The relationship of these graph classes are shown in Figure~\ref{fig:graph-class-relation}.

\begin{Thm}\label{thm:apex-minor-free}
\textbf{\cite{eppstein2000diameter}} Let $\mathcal{F}$ be a minor closed family of graphs. Then $\mathcal{F}$ has bounded local treewidth iff. $\mathcal{F}$ does not contain all apex graphs.
\end{Thm}

A graph is an apex graph if it has a vertex whose removal results in a planar graph. 
Theorem~\ref{thm:apex-minor-free} shows that a graph $G$ has bounded local treewidth iff. it is apex-minor-free.

\begin{Thm}\label{thm:linear-local-treewidth}
\textbf{\cite{demaine2004equivalence}} Any apex-minor-free graph has linear local treewidth.
\end{Thm}

\begin{figure}[h]
\vskip 0.2in
\begin{center}
\centerline{\includegraphics[width=\columnwidth]{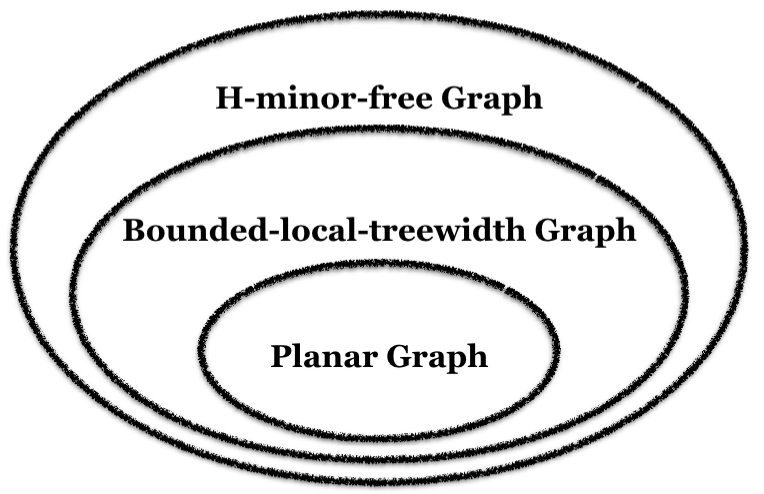}}
\caption{The class of planar graphs is contained in the class of bounded-local-treewidth graphs which is contained in the class of $H$-minor-free graphs.}
\label{fig:graph-class-relation}
\end{center}
\vskip -0.2in
\end{figure} 

\subsection{Clique sum}

The clique-sum operation is a way of combining two graphs by identifying their cliques. Suppose $G_1$ and $G_2$ are two graphs, $W_1 \subseteq V(G_1)$ and $W_2 \subseteq V(G_2)$ are two cliques of the same size. The clique sum of $G_1$ and $G_2$, denoted by $G_1 \oplus G_2$, is a graph by identifying $W_1$ and $W_2$ through a bijection $h: W_1 \to W_2$, and then possibly deleting some of the clique edges. The subgraph induced by the clique vertices in $G_1 \oplus G_2$ is called the \textit{join set}. The clique is called a $k$-sum if $|W_1|=|W_2|=k$, denoted by $G_1 \oplus_k G_2$. Since there are many possible bijections between vertices of $W_1$ and $W_2$, there are also many possible results for $G_1 \oplus G_2$.

The clique-sum operation plays an important role in the core of Robertson and Seymour's graph minor theory. The deep structural theorem~\cite{robertson2003graph} of graph minor theory states that any $H$-minor-free graph can be decomposed into a collection of graphs each of which can be embedded into a bounded-genus surface by deleting a bounded number of apex vertices where the number only depends on $V(H)$. These $h$-almost embedded graphs are combined in a tree structure by clique-sum operations. The clique-sum decomposition~\cite{demaine2005algorithmic,grohe2013asimple} is a building block by which the approximation algorithms for $H$-minor-free graphs can be achieved.

\subsection{Bounded-crossing-number graph}

We follow the definition of graphs with bounded crossings per edge in~\cite{grigoriev2007algorithms}.

\begin{Def}
An embedding of a graph $G$ on a surface $S$ of genus $g$ is a good embedding if all vertices of the graph are given as distinct points on $S$, no two edge crossings locate at the same point on $S$ and for any edge, no vertex except the endpoints of the edge locate on the edge.
\end{Def}


\begin{Def}
The \textit{crossing parameter} $\varphi$ of a graph $G$ embedded on a bounded-genus surface $S$ is the minimum over all good embeddings on $S$ of the maximum over all edges $e$ of the number of edge crossings of $e$.
\end{Def}

By the observation of~\cite{grigoriev2007algorithms}, the class of graphs with bounded crossing parameter is not minor-closed. Therefore it generalizes the discussions on $H$-minor-free graphs.

\section{Baker's Technique and Proof Sketch}
\label{section:baker}

Baker's technique, created over three decades ago, is a powerful tool for designing PTASs for NP-hard optimization problems on planar graphs. 
Its journal version is published in~\cite{baker1994approximation} that gives PTASs for many optimization problems like maximum independent set, minimum vertex cover, minimum dominating set,minimum edge dominating set, maximum triangle matching, maximum H-matching and maximum tile salvage.
To compute the maximum independent set (MIS) on a planar graph, it decomposes the planar embedding into several disjoint $k$-outerplanar graphs by removing all the vertices in layers congruent to $i~(\bmod~k+1)$ for some $0 \le i \le k$. Then the maximum independent set on each $k$-outerplanar graph can be computed by dynamic programming in $2^{O(k)}n$ time. 
The union of these maximum independent sets is a valid maximum independent set of $G$. The pigeonhole principle guarantees that there is at least one $i$ s.t. the final solution is at least $k/(k+1)$ optimal. Therefore, let $\epsilon = O(1/k)$, then maximum independent set on planar graphs can be approximated to $(1-\epsilon)$ optimal in $2^{O(k)}kn$ time.
In our model, computing the maximum independent set of $G$ is equivalent to computing $\max_{\sigma \in [0,1]^V} \sum_{i \in V} f_i$ by setting $f_i$ as
\begin{equation*}
f_i =
\begin{cases}
1,& \text{if } \sigma_i = 1, \text{ and } \sigma_j = 0 \text{ for } \forall j \text{ that } (i,j) \in E,\\
0,& \text{if } \sigma_i = 0,\\
-\infty,& \text{if } \sigma_i = 1, \text{ and } \sigma_j = 1 \text{ for } \exists j \text{ that } (i,j) \in E,
\end{cases}
\end{equation*}
where $\sigma_i = 1$ means vertex $i$ is in the independent set and $\sigma_i = 0$ otherwise. If there is an edge $(i,j) \in E$ and $\sigma_i = \sigma_j = 1$, then $f_i = f_j = -\infty$, which means the configuration is invalid.
Therefore if $\sum_{i \in V} f_i$ is maximized, we have the maximum independent set of $G$. 
Other combinatorial optimization problems like minimum vertex cover and minimum dominating set have similar encodings.

Then our task is extending Baker's technique from these specific potentials to a general version that allows $f_i$ to be any set of nonnegative functions derived from pairwise potentials.
To do this, the following lemma is a building block.

\begin{Lem}\label{lem:dynamic-programming}
Given a graph $G(V,E)$ with treewidth bounded by $k$, for any vertex set $U \subseteq V$ and any set of functions $f_i \ge 0$ derived from $\psi_i$ and $\psi_{ij}$ defined on $G$, $\max_{\sigma \in [q]^V}\sum_{i \in U} f_i$ and $\min_{\sigma \in [q]^V}\sum_{i \in U} f_i$ can be computed in $O(q^{O(k)}n)$ time.
\end{Lem}

The algorithm is a dynamic programming algorithm running on the tree decomposition of $G(V,E)$, denoted by $\mathcal{DP}_{G,\mathcal{F}}(U)$ where $\mathcal{F} = \{f_i: i \in V\}$ is the set of potentials attached on $V$.

We need to point out that our contribution is not inventing the dynamic programming algorithm of Lemma~\ref{lem:dynamic-programming} since there is no intrinsic difference between standard treewidth techniques such as junction tree algorithm.
Our contribution is changing the optimization goal from maximizing the sum of all potentials to maximizing the sum of potentials attached on a subset of vertices, which is significant for us to extend Baker's technique. 
Since the procedure of computation is different anyway, we elaborate the proof of Lemma~\ref{lem:dynamic-programming} as follows to clarify the details.

\begin{proof}[Proof of Lemma~\ref{lem:dynamic-programming}]
Since the treewidth of $G$ is bounded by $k$, we use the algorithm of~\cite{bodlaender1993linear} to construct a tree decomposition $T = (I,F)$ rooted at $r \in I$ with treewidth $k$ for $G$ in linear time (but exponential in $k$).
For each $i \in I$, the subtree of $T$ rooted at $i$ is denoted by $T_i$.
The set of vertices in $T_i$ is denoted by $V_{T_i}$.
The configurations of bag $X_i$ where $i \in I$ is denoted by $\sigma_{X_i}$.
Suppose the child nodes of $i \in I$ are $i_1, \ldots, i_d \in I$ and the parent node of $i \in I$ is $p_i \in I$.

The dynamic programming runs from the leaves to the root. We enumerate the all the possible configurations of $X_i - X_{p_i}$ for each bag $X_i$. For root $r$, $X_r - X_{p_r} = X_r$. By the definition of tree decomposition, $X_{i_x} \cap X_{i_y} \subseteq X_i$ for $1 \le x \ne y \le d$. Therefore, $X_{i_t} - X_{i}$ for $1 \le t \le d$ are pairwise disjoint.
Let $S_{i \backslash p_i}^U(\sigma_{i \backslash p_i})$ denote the max-sum of the $f_i$ attached on vertices in $(U \cap V_{T_i}) - (X_{p_i}\cup\partial X_{p_i})$ with the configurations of vertices in $X_i - X_{p_i}$ being fixed to $\sigma_{i \backslash p_i}$.
The set $\partial X_{p_i}$ denotes the vertices adjacent to vertices in $X_{p_i}$ but not in $X_{p_i}$.
Note that $S_{i \backslash p_i}^U(\sigma_{i \backslash p_i})$ does not include the sum of $f_i$ attached on vertices in $\partial X_{p_i}$. This is because their values are not fixed when the configurations of $X_{p_i}$ are not given.
The value of $S_{i \backslash p_i}^U(\sigma_{i \backslash p_i})$ can be computed by the following recurrence:
\begin{equation*}
S_{i \backslash p_i}^U(\sigma_{i \backslash p_i}) =
\max_{\{\sigma_{i_t \backslash i: 1 \le t \le d}\}}\bigg\{ \Gamma_{ X_i - X_{p_i}}^{\sigma_{i \backslash p_i}} + \sum_{t=1}^d S_{i_t \backslash i}^U(\sigma_{i_t \backslash i}) \bigg\}
\end{equation*}
where $\Gamma_{ X_i - X_{p_i}}^{\sigma_{i \backslash p_i}}$ is the sum of $f_i$ attached on vertices in $(U \cap (X_i\cup\partial X_i)) - (X_{p_i}\cup\partial X_{p_i})$ when the configuration of $X_i - X_{p_i}$ is fixed to $\sigma_{i \backslash p_i}$.
Then
$\max_{\sigma \in [q]^V}\sum_{i \in U} f_i = \max_{\sigma_{r \backslash p_r}} S_{r \backslash p_r}^U(\sigma_{r \backslash p_r})$.
Note that this recursion only holds for the set of $f_i$ derived from the energy functions, which takes the form of $f_i = \psi_i + \sum_{j \in N(i)} \alpha_{ij} \cdot \psi_{ij}$.
Since $\psi_{ij}$ is pairwise so that we can enumerate all the possibilities of each $\sigma_{i_t \backslash i}$ for each $1 \le t \le d$ respectively.
The min-sum can be computed in the same way.
For each bag $X_i$, since $|X_i| \le k+1$ for all $i \in I$, $\sigma_{i \backslash p_i}$ has at most $q^{k+1}$ possible values. We compute $S_{i \backslash p_i}^U(\sigma_{i \backslash p_i})$ for each $\sigma_{i \backslash p_i}$ at most once.
To compute $S_{i \backslash p_i}^U(\sigma_{i \backslash p_i})$, we need to know $X_i \cap X_{p_i}$ and $X_i \cap \partial X_{p_i}$ for each $X_i$. 
Given the tree decomposition, $X_i \cap X_{p_i}$ can be preprocessed in $O(kn)$ time if the vertices in each $X_i$ are stored in order or using data structures such as hash tables. Given $X_i \cap X_{p_i}$, $X_i \cap \partial X_{p_i}$ can be computed in at most $O(k^2 n)$ time.
Therefore, the total time complexity of our dynamic programming algorithm is $O(q^{O(k)}n)$.
\end{proof}

Similar to the original Baker's technique, our approximation algorithm follows a divide-and-conquer style.
The vertices in the input graph $G$ are labeled by numbers from $0$ to $k+1$. 
The vertices labeled by the same number belong to the same level.
The vertices of level $i$ are only adjacent to vertices of level $i-1~(\bmod~k+2)$ and vertices of level $i+1~(\bmod~k+2)$.
Then we delete 0 or several levels of vertices to decompose the input graph into several disjoint subgraphs and use the dynamic programming algorithm of Lemma~\ref{lem:dynamic-programming} to compute a partial solution on each subgraphs. 
Finally, we combine these partial solutions to obtain a approximation of the optimal solution.
For different graph classes, the ways of vertex labeling are different. Different ways of vertex labelling correspond to different graph decomposition techniques, which will be specified in Section~\ref{section:decomposition} for all graph classes mentioned in Theorem~\ref{thm:main-energy}.

More specifically, we decompose the input graph by deleting all the edges and vertices (if exist) with labels $\ell$ that $i < \ell < i+\Delta~(\bmod~k+2)$ $(0 \le i \le k+1)$ where $\Delta \ge 1$ is a constant depending on the input graph.
It satisfies that after the deletion $G$ is decomposed into several subgraphs whose treewidths are bounded by $O(k)$. 
The subgraphs are denoted by $G_1^i, \ldots, G_t^i$. The vertices adjacent to deleted edges in subgraph $G_j^i(V_j^i, E_j^i)$ $(0 \le j \le t)$ are called \textit{boundary nodes}, denoted by $B_j^i$. The non-boundary nodes are denoted by $A_j^i = V_j^i - B_j^i$. 
Then we use $\mathcal{DP}_{G_j^i,\mathcal{F}_j^i}(A_j^i)$ to maximize the sum of potentials attached on vertices in $A_j^i$ while ignoring the values of the potentials attached on vertices in $B_j^i$. 
Actually, after the edge deletion, the outputs of the potentials attached on vertices in $B_j^i$ are undefined since they cannot read all the inputs. 
The configuration of vertices in $B_j^i$ is only required for calculating the values of potentials attached on vertices in $\partial B_j^i \cap A_j^i$ where $\partial B_j^i $ denotes the vertices adjacent to $B_j^i$.
When the sum of the potentials attached on $A_j^i$ has been calculated by $\mathcal{DP}_{G_j^i,\mathcal{F}_j^i}(A_j^i)$, the configuration of $B_j^i$ is also fixed.
 
Suppose $A^i = \bigcup_{j=1}^t A_j^i$ and $B^i = \bigcup_{j=1}^t B_j^i$.
Let $S_{A^i} = \sum_{j=1}^t S_{A_j^i}$ where $S_{A_j^i}$ is the sum of the potentials attached of vertices in $A_j^i$ calculated by $\mathcal{DP}_{G_j^i,\mathcal{F}_j^i}(A_j^i)$. 
Similarly, let $S_{B^i} = \sum_{j=1}^t S_{B_j^i}$.
Suppose $S_{\text{OPT}}$ is the optimum of $E(\sigma)$.
By the pigeonhole principle, for at least one $i$, at most $\frac{\Delta+1}{k+2}$ of $S_{\text{OPT}}$ is produced by potentials attached on vertices on $V - A^i$. 
Therefore, it holds that
$S_{A^i} = \sum_{j=1}^t S_{A_j^i} \ge \left(1-\frac{\Delta+1}{k+2}\right)\cdot S_{\text{OPT}}$.
Since $f_i \ge 0$ for all $i \in V$, thus we have
$S = S_{A^i} + S_{B^i} + S_{\Delta} \ge S_{A^i} \ge \left(1-\frac{\Delta+1}{k+2}\right)\cdot S_{\text{OPT}}$
where $S$ is the solution computed by our approximation algorithm and $S_{\Delta}$ is the sum of potentials attached on the vertices in $V - A^i - B^i$.
Given a fixed error $0 < \epsilon < 1$, it needs to satisfy that $\frac{\Delta+1}{k+2} \le \epsilon$, which implies $k \ge \left\lceil \frac{\Delta+1}{\epsilon}-2 \right\rceil = O(\frac{1}{\epsilon})$. 

As the running time of $\mathcal{DP}_{G_j^i,\mathcal{F}_j^i}(A_j^i)$ is $O(q^{O(k)}n)$ and  the dynamic programming for different $i$ can be computed in parallel, the time complexity for a fixed $i$ is $O(q^{O(k)}n)$. 
For each $0 \le i \le k+1$, we need to repeat the dynamic programming. The total time complexity is $O(q^{O(k)}kn)$.
This completes the proof sketch of Theorem~\ref{thm:main-energy}.

\section{Graph Decompositions}
\label{section:decomposition}

\subsection{Planar Graphs}

For planar graphs, $\Delta = 1$. Given a planar embedding of a planar graph $G$, we decompose it into several disjoint $(k+2)$-outerplanar subgraphs $G_1^i, \ldots, G_t^i$ by deleting all the edges between levels congruent to $i~(\bmod~k+2)$ and $i+1~(\bmod~k+2)$ for some integer $i$ that $0 \le i \le k+1$. 
As we have proved in Section~\ref{section:baker}, the result is $\left(1-\frac{\Delta+1}{k+2}\right) = \frac{k}{k+2}$ optimal.

\subsection{Bounded-local-treewidth Graphs}

For bounded-local-treewidth graphs, $\Delta = 1$.
Choosing any vertex $v \in V$ as root, construct a BFS tree $T$ rooted at $v$. The layer of vertices is defined as its distance to $v$. Moreover, the set of vertices from layer $i$ to layer $j$ is denoted by 
$L_v^G[i,j] = \{u \in V | i \le \text{dist}(u,v) \le j\}$.
If $i > j$, $L_v^G[i,j] = \emptyset$.
For any $i \le j$, $L_v^G[i,j]$ has bounded local treewidth. This is because if we obtain a minor $H$ of $G$ by contracting the subgraph of $G$ induced by $L_v^G[0,i-1]$ to a single vertex $v'$, $L_v^G[i,j] \subseteq L_{v'}^H[1,j-i+1]$. Since $G$ is apex-minor-free, $H$ is also apex-minor-free. Therefore, $H$ has bounded local treewidth. Then we have $tw(L_v^G[i,j]) = O(j-i+1)$. It implies any subgraph induced by consecutive $k$ levels of vertices in $T$ has treewidth bounded by $O(k)$.
We delete all the edges between levels congruent to $i~(\bmod~k+2)$ and $i+1~(\bmod~k+2)$ for some integer $i$ that $0 \le i \le k+1$. Then $G$ is decomposed into several disjoint subgraphs $G_1^i,\ldots,G_t^i$. 
Hence the result is $\left(1-\frac{\Delta+1}{k+2}\right) = \frac{k}{k+2}$ optimal.

\subsection{$H$-minor-free Graphs}

For $H$-minor-free graphs, $\Delta = 1$.
A graph $H$ is a $k$-apex graph of a graph $G$ if $G = H \backslash A$ for some subset $A$ of at most $k$ vertices which is called \textit{apices}.
The definition of almost-embeddable graph is given as follows.

\begin{Def}
A graph $G$ is \textit{almost-embeddable} on a surface $\Sigma$ if $G$ can be written as the union of $k+1$ graphs $G_0 \cup G_1 \cup \ldots \cup G_k$, satisfying the following conditions:
(1) $G_0$ has an embedding on $\Sigma$.
(2) The graphs $G_1,G_2,\ldots,G_k$ are pairwise disjoint, called vortices.
(3) For each index $i \ge 1$, there is a disk $D_i$ inside some face $F_i$ of $G_0$, such that $U_i = V(G_0) \cap V(G_i) = V(G_0) \cap D_i$. Moreover, the disks $D_i$ are pairwise disjoint.
(4) For each index $i \ge 1$, the subgraph $G_i$ has pathwidth less than $k$. Moreover, $G_i$ has a path decomposition $\langle X_i^1,X_i^2,\ldots,X_i^{r_i} \rangle$ with $r_i \le k$, such that $v_i^j \in X_i^j$ for $1 \le j \le r_i$, where $v_i^1,v_i^2,\ldots,v_i^{r_i}$ are the vertices of $U_i$ indexed in cyclic order around the face $F_i$, clockwise or anti-clockwise.
\end{Def}


\begin{Lem}\label{lem:almost-embeddable}
\textbf{\cite{grohe2003local}} The class of all graphs almost embeddable in a fixed surface $S$ has linear local treewidth.
\end{Lem}

\begin{Def}
A graph $G$ is \textit{$h$-almost-embeddable} on a surface $\Sigma$ if $H$ is a $h$-apex graph of a graph that is almost embeddable on $\Sigma$.
\end{Def}

\begin{Thm}\label{thm:rs2003}
\textbf{\cite{robertson2003graph}} For any graph $H$, there is an integer $h \ge 0$ depending only on $|V(H)|$ such that any $H$-minor-free graph is a $h$-clique sum of a finite number of graphs that are $h$-almost-embeddable on some surfaces on which $H$ cannot be embedded.
\end{Thm}

Theorem~\ref{thm:rs2003} says that any $H$-minor-free graph $G$ can be expressed as a ``tree structure" of pieces, where each piece can be embedded on a surface on which $H$ cannot be embedded after deleting at most $h$ apex vertices.

\begin{Thm}\label{devos-clique-sum}
\textbf{\cite{devos2004excluding}} For the clique-sum decomposition of a $H$-minor-free graphs, written as $G_1 \oplus \ldots \oplus G_t$, the join set of each clique-sum operation between $G_1 \oplus \ldots \oplus G_{i-1}$ and $G_i$ is a subset of the apices of $G_i$. Moreover, each join set of the clique-sum decomposition involving $G_j$ contains at most three vertices of the bounded-genus part of $G_j$.
\end{Thm}

The following theorem gives a polynomial-time algorithm for computing the clique-sum decomposition with the additional properties guaranteed by Theorem~\ref{devos-clique-sum}.

\begin{Thm}\label{demaine-clique-sum}
\textbf{\cite{demaine2005algorithmic}} For a fixed graph $H$, there is a constant $c_H$ such that, for any integer $k \ge 1$ and for every $H$-minor-free graph $G$, the vertices of $G$ (or the edges of $G$) can be partitioned into $k + 1$ sets such that any $k$ of the sets induce a graph of treewidth at most $c_H k$. Furthermore, such a partition can be found in polynomial time.
\end{Thm}

Grohe et al.~\cite{grohe2013asimple} give a quadratic time algorithm that is faster for computing the clique-sum decomposition of $H$-minor-free graphs. When we describe our approximation algorithm, we always assume that such a clique-sum decomposition has already been given.

\begin{Def}
Graph class $\mathcal{G}$ has truly sublinear treewidth with parameter $\lambda$ where $0 < \lambda < 1$, if for every $\eta > 0$, there exists $\beta > 0$ such that for any graph $G \in \mathcal{G}$ and $X \subseteq V(G)$ the condition $tw(G \backslash X) \le \eta$ yields that $tw(G) \le \eta + \beta|X|^\lambda$.
\end{Def}

\begin{Lem}\label{lem:truly-sublinear-treewidth}
\textbf{\cite{fomin2011bidimensionality}} Let $\mathcal{G}_H$ be a class of graphs excluding a fixed graph $H$ as a minor, then $\mathcal{G}_H$ has truly sublinear treewidth with $\lambda = \frac{1}{2}$.
\end{Lem}

Our algorithm leverages the graph decomposition technique in~\cite{demaine2005algorithmic}.
Suppose the clique-sum decomposition of the input $H$-minor-free graph $G$ is $G_1 \oplus \ldots \oplus G_t$ where each $G_i$ ($1 \le i \le t$) is an $h$-almost embeddable graph. The join set $J_i$ of the $i$-th clique-sum operation $(G_1 \oplus \ldots \oplus G_i) \oplus G_{i+1}$ is a subset of the apex set $X_{i+1}$ of $G_{i+1}$. Our approximation algorithm takes the clique-sum decomposition as input, the apex set $X_i$ of each $G_i$ is given as part of the input clique-sum decomposition.
By the definition of the $h$-almost embeddable graphs, $G_i - X_i$ is almost embeddable on a bounded-genus surface where $X_i$ contains at most $h$ vertices. By lemma~\ref{lem:almost-embeddable}, $G_i - X_i$ has bounded local treewidth.
From $i = 1$ to $t$, we choose a vertex $v_i \in G_i - X_i$ and construct a BFS tree $T_i$ of $G_i - X_i$ rooted at $v_i$. Each vertex in $u \in G_i - X_i$ is labeled by the distance between $u$ and $v_i$ modulo $k+2$. After this step, we delete all the edges between levels labeled by $\ell~(\bmod~k+2)$ and the adjacent levels labeled by $\ell+1~(\bmod~k+2)$. Then the $G_i - X_i$ part is decomposed into several disjoint subgraphs with treewidth at most $c(k+2)$ for some constant $c > 0$.
Since $J_{i-1} \subseteq X_{i}$, the vertices in $J_{i-1}$ has already been labeled in $G_1 \oplus G_2 \oplus \ldots \oplus G_{i-1}$.
We label the vertices in $X_i \backslash J_{i-1}$ arbitrarily by the integers from 0 to $k+1$.
After the edge deletions, the obtained graphs $G_i'$ are still $H$-minor-free for $1 \le i \le t$.
By~\cite{fomin2011bidimensionality}, the treewidth of $G_i'$ is at most $ck + \beta |X_i|^{1/2} \le ck + \beta h^{1/2} = O(k)$.
It is known in~\cite{demaine2004approximation} that $tw(G \oplus H) \le \max\{tw(G),tw(H)\}$, thus we have
$tw(G_1' \oplus G_2' \oplus \ldots \oplus G_t') \le \max\{tw(G_1'), \ldots, tw(G_t')\}$.
This shows that any $H$-minor-free graph can be transformed into a graph with treewidth bounded by $O(k)$ by deleting at most $|E|/k$ edges. 
Given a clique-sum decomposition, the vertex labeling and edge deletions can be finished in linear time.
By a similar argument, we achieve the $\left(1-\frac{\Delta+1}{k+2}\right) = \frac{k}{k+2}$ optimal solution in $O(q^{O(k)}kn)$ time.

\subsection{Bounded-crossing-number Graphs}

For bounded-crossing-number graphs, $\Delta = \varphi+1$.
We obtain a planar graph $G'=(V',E')$ by replacing each edge crossing of $G$ by a new vertex. 
Construct a breadth first search tree $T$ of $G'$, rooted at any $v \in V'$.
The level of a vertex is defined as the distance from the vertex to the root $v$ of $T$.

For each level $i~(\bmod~\Lambda)$ in $T$ where $\Lambda = (\varphi+1)(k+2)$, we remove the levels from $i+1~(\bmod~\Lambda)$ to $i+\varphi~(\bmod~\Lambda)$ of $G'$. 
Then $G'$ is decomposed into several subgraphs $\mathcal{H}_i = \{H_1^i,\ldots,H_t^i\}$, where each $H_j^i = (N_j^i,E_j^i)$ that $1 \le j \le t$ 
contains at least $k+2-\varphi$ levels of $G$.
Let $V_j^i = N_j^i \cap V$ and $G_j^i = G[V_j^i]$ that represents the subgraph of $G$ induced by $V_j^i$. 
Since the number of crossings per edge is at most $\varphi$ and $\varphi$ consecutive levels of vertices are removed from $G'$, thus after the removal all the subgraphs $G_j^i$ are disjoint with each other.
%

By an observation of~\cite{grigoriev2007algorithms}, it satisfies that $tw(G_j^i) \ge 2 \cdot tw(H_j^i)+1$.
Since $H_j^i$ is embeddable on a bounded-genus surface and it is well-known that such graphs have linear local treewidth.
Thus we are able to deduce that each $G_j^i$ has treewidth $O(\varphi k) = O(k)$.
Therefore, by a similar argument, we cab achieve a $\left(1-\frac{\Delta+1}{k+2}\right) = \frac{k-\varphi}{k+2}$ optimal solution in $O(q^{O(k)}kn)$ time.

\section{Conclusion}

In this paper, we give EPTAS for energy maximization on planar graphs, bounded-local-treewidth graphs, $H$-minor-free graphs and bounded-crossing-number graphs. We also prove the inapproximability property for energy minimization. A clearer characterization for the complexity of energy minimization can be left as future research.





\end{document}